\documentclass[letterpaper, 11pt]{article} 
\usepackage[top=1.0in, bottom=1.0in, left=1.0in, right=1.0in]{geometry}
\usepackage{amsmath, amsfonts, amsthm, amssymb}
\usepackage{url}

\newcommand{\E}{\mathbb{E}}

\newcommand{\R}{\mathbb{R}}

\newcommand{\N}{\mathbb{N}}
\newcommand{\C}{\mathcal{C}}

\renewcommand{\P}{\mathcal{P}}

\newcommand{\eps}{\epsilon}

\renewcommand\vec[1]{\ensuremath\boldsymbol{#1}}

\newcommand\pref[1]{#1}

\renewcommand{\paragraph}[1]{\medskip\noindent{\bf #1}}

\makeatletter
\newtheorem*{rep@theorem}{\rep@title}
\newcommand{\newreptheorem}[2]{%
\newenvironment{rep#1}[1]{%
  \def\rep@title{#2 \ref{##1}}%
  \begin{rep@theorem}}%
  {\end{rep@theorem}}}
\makeatother

\newtheorem{theorem}{Theorem}

\newtheorem{corollary}[theorem]{Corollary}
\newtheorem{lemma}[theorem]{Lemma}
\newtheorem{claim}[theorem]{Claim}

\newtheorem{definition}[theorem]{Definition}

\newreptheorem{theorem}{Theorem}
\newreptheorem{proposition}{Proposition}
\newreptheorem{corollary}{Corollary}
\newreptheorem{lemma}{Lemma}
\newreptheorem{claim}{Claim}
\newreptheorem{conjecture}{Conjecture}

\theoremstyle{definition}

\begin{document}

\begin{titlepage}

\title{Stronger Impossibility Results for Strategy-Proof Voting with i.i.d.~Beliefs}
\setcounter{page}{0}

\author{Samantha Leung\thanks{Leung was supported in part by NSF grants IIS-0911036 and CCF-1214844, AFOSR grant FA9550-08-1-0438, ARO grant W911NF-14-1-0017, and by the DoD Multidisciplinary University Research Initiative (MURI) program administered by AFOSR under grant FA9550-12-1-0040.} \\ Cornell University\\ samlyy@cs.cornell.edu \and
Edward Lui \\ Cornell University \\ luied@cs.cornell.edu \and
Rafael Pass\thanks{Pass is supported in part by an Alfred P. Sloan Fellowship, Microsoft New Faculty Fellowship, NSF CAREER Award CCF-0746990, NSF Award CCF-1214844, NSF Award CNS-1217821, AFOSR YIP Award FA9550-10-1-0093, and DARPA and AFRL under contract FA8750-11-2-0211.} \\ Cornell University \\ rafael@cs.cornell.edu}

\maketitle

\begin{abstract}
The classic Gibbard-Satterthwaite theorem says that every strategy-proof voting rule with at least three possible candidates must be dictatorial. In \cite{McL11}, McLennan showed that a similar impossibility result holds even if we consider a weaker notion of strategy-proofness where voters believe that the other voters' preferences are i.i.d.~(independent and identically distributed): If an anonymous voting rule (with at least three candidates) is strategy-proof w.r.t.~all i.i.d.~beliefs and is also Pareto efficient, then the voting rule must be a random dictatorship. In this paper, we strengthen McLennan's result by relaxing Pareto efficiency to $\epsilon$-Pareto efficiency where Pareto efficiency can be violated with probability $\epsilon$, and we further relax $\epsilon$-Pareto efficiency to a very weak notion of efficiency which we call $\epsilon$-super-weak unanimity. We then show the following: If an anonymous voting rule (with at least three candidates) is strategy-proof w.r.t.~all i.i.d.~beliefs and also satisfies $\epsilon$-super-weak unanimity, then the voting rule must be $O(\epsilon)$-close to random dictatorship.
\end{abstract}
\end{titlepage}

\section{Introduction}

People have long desired to have a good voting rule that is \emph{strategy-proof}---that is, the voters would not want to lie about their true preferences. Unfortunately, the celebrated Gibbard-Satterthwaite theorem~\cite{Gib73,Sat75} shows that if there are at least three possible candidates, then any deterministic strategy-proof voting rule has to be dictatorial---that is, there exists a fixed voter whose top choice is always the winner. Although the Gibbard-Satterthwaite theorem only applies to \emph{deterministic} voting rules, Gibbard later generalized the Gibbard-Satterthwaite theorem to \emph{randomized} voting rules \cite{Gib77}. In particular, Gibbard showed that any randomized strategy-proof voting rule has to be a probability distribution over \emph{unilateral rules} and \emph{duple rules}, where a unilateral rule depends only on a single voter, and a duple rule chooses only between two possible candidates; furthermore, if the voting rule satisfies the natural condition of \emph{Pareto efficiency}---that is, the voting rule \emph{never} chooses a candidate $y$ that is dominated by some other candidate $x$ by every voter---then the voting rule must be a probability distribution over dictatorial voting rules.

The notion of strategy-proofness, however, is quite strong. It requires voters to truthfully report their preferences, \emph{no matter} what preferences the other voters have (and in particular, even if the voter \emph{knows} exactly the preferences of everyone else). One may thus hope that these impossibility results can be circumvented by relaxing this requirement. For instance, for the case of ``large-scale'' voting, it makes sense to assume that each voter has some belief about the preferences of the other voters, and additionally that these preferences are independent and identically distributed (i.i.d.)---we refer to such a notion as \emph{strategy-proofness w.r.t.~(with respect to) i.i.d.~beliefs}. Unfortunately, this weakening does not make things much better: A recent result by McLennan \cite{McL11} shows that if an \emph{anonymous}\footnote{A voting rule is anonymous if it does not depend on the identity of the voters.} voting rule (with at least three candidates) is strategy-proof w.r.t.~\emph{all} i.i.d.~beliefs and is also Pareto efficient, then the voting rule must be a random dictatorship---that is, a uniformly random voter's top choice is chosen as the winner. 

Pareto efficiency, however, is a strong condition. When dealing with \emph{randomized} voting rules, a natural relaxation (borrowing from the literature on randomized algorithms or cryptography) would be to allow Pareto efficiency to be violated with some ``tiny'' probability $\eps$. For instance, if this probability $\eps$ is exponentially small in the number of voters, then it seems that the voting rule is still perfectly reasonable. Unfortunately, we show that relaxing Pareto efficiency to $\eps$-Pareto efficiency (where Pareto efficiency can be violated with probability $\eps$) does not help, even for rather large values of $\eps$, and even for a significantly weaker notion of efficiency which we call $\eps$-\emph{super-weak unanimity}: $\eps$-super-weak unanimity requires that for every candidate $x$, there exists \emph{some} preference profile with $x$ being the top choice of every voter, such that the voting rule chooses $x$ with probability at least $1-\eps$.

\begin{theorem}[Informal]
Suppose there are at least three candidates. Let $v$ be any anonymous randomized voting rule that is strategy-proof w.r.t.~all i.i.d.~beliefs, and satisfies $\eps$-super-weak unanimity. Then, $v$ is $O(\eps)$-close to the random dictatorship voting rule.
\end{theorem}

Thus, even for an extremely weak notion of what it means to be a ``reasonable'' voting rule, strategy-proofness w.r.t.~all i.i.d.~beliefs cannot be achieved.

\section{Preliminaries}
\label{sec:preliminaries}

Given an integer $k \in \N$, let $[k] = \{1, \ldots, k\}$. Let $\C$ be any finite set of \emph{candidates} (or \emph{alternatives}). A \emph{preference ordering} on $\C$ is a strict total order on the set of candidates $\C$; let $\P$ denote the set of all preference orderings on $\C$. Given a subset $A \subseteq \C$ of candidates, let $L(A)$ denote the set of preference orderings (i.e., strict total orders) on $A$. Given a preference ordering $P$ and a pair of candidates $x,y \in \C$, we shall write $x \pref{P} y$ to mean that \emph{$x$ is (strictly) preferred over $y$ in $P$}, i.e., $x$ is ranked higher than $y$ according to $P$. Given a preference ordering $P$, let $top(P)$ denote the highest-ranked candidate according to $P$, i.e., $top(P)$ is the candidate $x$ in $\C$ such that $x \pref{P} y$ for every $y \in \C \setminus \{x\}$. 

Throughout this paper, we will use $n$ to denote the number of voters, and $m$ to denote the number of candidates in $\C$; we will often treat $m$ as a constant. A \emph{preference profile} is a vector of length $n$ whose components are preference orderings in $\P$; that is, a preference profile is simply an element of $\P^n$ which specifies the (submitted) preference orderings of $n$ voters. 

Given a finite set $S$, let $\Delta(S)$ denote the set of all probability distributions over $S$. A (randomized) \emph{voting rule} is a function $v: \P^n \to \Delta(\C)$ (or $v: \P^* \to \Delta(\C)$ if $v$ works for any number of voters) that maps preference profiles to probability distributions over candidates; intuitively, $v(\vec{P})$ is a distribution over $\C$ that specifies the probability that each candidate is selected when the submitted votes form the preference profile $\vec{P}$. A voting rule $v$ is said to be \emph{deterministic} if for every preference profile $\vec{P}$, the distribution $v(\vec{P})$ assigns probability 1 to some candidate. A voting rule $v$ is said to be \emph{anonymous} if $v$ does not depend on the order in which the preference orderings appear in the input, i.e., $v(P_1, \ldots, P_n) = v(P_{\sigma(1)}, \ldots, P_{\sigma(n)})$ for every preference profile $(P_1, \ldots, P_n) \in \P^n$ and every permutation $\sigma: [n] \to [n]$. In this paper, we will only consider anonymous voting rules; most common voting rules are indeed anonymous, and one can argue that anonymous voting rules are more fair and democratic than non-anonymous ones. 

Given a (randomized) voting rule $v: \P^n \to \Delta(\C)$, a candidate $x \in \C$, and a preference profile $\vec{P}$, let $v(x,\vec{P})$ be the probability mass assigned to $x$ by the distribution $v(\vec{P})$; we also refer to $v(x,\vec{P})$ as the \emph{selection probability of $x$ with respect to $v$ and $\vec{P}$}, since $v(x,\vec{P})$ is the probability that candidate $x$ is selected by the voting rule $v$ when the input preference profile is $\vec{P}$. A \emph{utility function} is a function $u: \C \to [0,1]$ that assigns a real number in $[0,1]$ to each candidate in $\C$.\footnote{It is not important that the codomain of the utility function $u$ is $[0,1]$; as long as the codomain is bounded, the results of this paper still hold with minor modifications.} Given a preference ordering $P$ and a utility function $u$, we say that $u$ is \emph{consistent with $P$} if for every pair of candidates $x,y \in \C$, we have $u(x) > u(y)$ if and only if $x \pref{P} y$.

A voting rule is \emph{Pareto efficient} if it never chooses a Pareto dominated candidate, i.e., a candidate $y$ such that all the voters prefer $x$ over $y$ for some candidate $x$. A slight relaxation of Pareto efficiency is \emph{$\eps$-Pareto efficiency}, where we allow the voting rule to choose a Pareto dominated candidate with probability at most $\eps$.

\begin{definition}[$\eps$-Pareto efficiency]
A voting rule $v: \P^n \to \Delta(\C)$ is \emph{$\eps$-Pareto efficient} if for every pair of candidates $x,y \in \C$ and every preference profile $\vec{P} = (P_1, \ldots, P_n) \in \P^n$ such that $x \pref{P_i} y$ for every $i \in [n]$, we have $v(y,\vec{P}) \leq \eps$. 
\end{definition}

The \emph{random dictatorship} voting rule is the voting rule $v_{dict}$ that chooses a voter uniformly at random and then chooses her top choice, i.e., for every preference profile $\vec{P} = (P_1, \ldots, P_n) \in \P^n$ and every candidate $x \in \C$, we have $v_{dict}(x,\vec{P}) = \frac{|\{i \in [n] : top(P_i) = x\}|}{n}$. 

See \ref{app:Gibbard-Satterthwaite} for background information on the Gibbard-Satterthwaite theorem \cite{Gib73,Sat75} and Gibbard's generalization of the Gibbard-Satterthwaite theorem to randomized voting rules \cite{Gib77}.

\subsection{Strategy-Proofness with respect to a Set of Beliefs} 

Gibbard's generalization \cite{Gib77} of the Gibbard-Satterthwaite theorem shows that when there are at least three candidates, we cannot even construct good \emph{randomized} voting rules that are strategy-proof. Given this impossibility result, let us consider relaxed notions of strategy-proofness. We observe that strategy-proofness requires that no voter would want to lie about her true preference even if the voter \emph{knows} the submitted preferences of \emph{all} the other voters. However, in many realistic scenarios, a voter is \emph{uncertain} about how other voters will vote, and she would only lie if she \emph{believes} that she can gain utility in expectation by lying. As a result, we consider a relaxed notion of strategy-proofness where we consider the voter's \emph{belief} of how the other voters will vote. The standard notion of strategy-proofness requires that no voter would want to lie regardless of what her belief is. To weaken the notion of strategy-proof, one can require that no voter would want to lie as long as her belief belongs in a certain set of beliefs. Let us now move to formalizing these notions.

In this paper, we will only consider beliefs that are \emph{i.i.d.}~(independent and identically distributed), meaning that for each belief, the other voters' preference orderings are sampled independently from some distribution $\phi$ over preference orderings. Thus, for simplicity, we define a \emph{belief} to be a probability distribution over the set $\P$ of preference orderings, representing a voter's belief that each of the other voters will have a preference ordering drawn independently from this distribution. We now state the definition of \emph{strategy-proof with respect to a set of beliefs}.

\begin{definition}[Strategy-proof w.r.t.~a set $\Phi$ of beliefs]
A voting rule $v: \P^n \to \Delta(\C)$ is \emph{strategy-proof w.r.t.~a set $\Phi$ of beliefs} if for every $i \in [n]$, every pair of preference orderings $P_i, P_i' \in \P$, every belief $\phi \in \Phi$, and every utility function $u_i$ that is consistent with $P_i$, we have
\begin{align*}
  \E[u_i(v(\vec{P}_{-i},P_i))] \geq \E[u_i(v(\vec{P}_{-i},P_i'))],
\end{align*}
where $\vec{P}_{-i} \sim \phi^{n-1}$.
\end{definition}

\section{Impossibility of Strategy-Proofness w.r.t.~all i.i.d.~Beliefs}
\label{sec:impossibility}

In this section, we prove our impossibility result that says that if there are at least three candidates, then it is not possible to construct a voting rule that is strategy-proof for all i.i.d.~beliefs and satisfies $\eps$-super-weak unanimity, unless the voting rule is $O(\eps)$-close to being the random dictatorship voting rule. We begin with some definitions. We call a voting rule \emph{weakly strategy-proof} if it is strategy-proof with respect to the set of all i.i.d.~beliefs.

\begin{definition}[Weakly strategy-proof]
A voting rule $v: \P^n \to \Delta(\C)$ is \emph{weakly strategy-proof} if $v$ is strategy proof with respect to the set of all i.i.d.~beliefs.
\end{definition}

Suppose all the voters have the same top choice, say candidate $x$; then, we would expect the voting rule to choose $x$ as the winner. We call this property \emph{strong unanimity}; again, we slightly relax this property to \emph{$\eps$-strong unanimity}, where we allow the voting rule to choose the common top candidate with probability at least $1-\eps$ instead of probability $1$. 

\begin{definition}[$\eps$-strong unanimity]
A voting rule $v: \P^n \to \Delta(\C)$ satisfies \emph{$\eps$-strong unanimity} if for every candidate $x \in \C$ and every preference profile $\vec{P} = (P_1, \ldots, P_n) \in \P^n$ such that $top(P_i) = x$ for every $i \in [n]$, we have $v(x,\vec{P}) \geq 1 - \eps$. 
\end{definition}

It is easy to see that strong unanimity is weaker than Pareto efficiency (modulo a factor of $m$ for the $\eps$ version of the properties). Now, suppose all the voters have the exact same preference ordering $P$; then, we would expect the voting rule to choose the top candidate of $P$. We call this property \emph{weak unanimity}; again, we state an $\eps$ version of the definition.

\begin{definition}[$\eps$-weak unanimity]
A voting rule $v: \P^n \to \Delta(\C)$ satisfies \emph{$\eps$-weak unanimity} if for every preference ordering $P \in \P$ and every preference profile $\vec{P} = (P_1, \ldots, P_n) \in \P^n$ such that $P_i = P$ for every $i \in [n]$, we have $v(top(P),\vec{P}) \geq 1 - \eps$. 
\end{definition}

It is clear that $\eps$-weak unanimity is weaker than $\eps$-strong unanimity. We finally define \emph{$\eps$-super-weak unanimity}, which is \emph{even weaker} than $\eps$-weak unanimity, and requires that for every candidate $x \in \C$, there exists some preference profile $\vec{P}$ with $x$ at the top of every preference ordering in $\vec{P}$, such that the voting rule on $\vec{P}$ will choose $x$ as the winner with probability at least $1-\eps$. 

\begin{definition}[$\eps$-super-weak unanimity]
A voting rule $v: \P^n \to \Delta(\C)$ satisfies \emph{$\eps$-super-weak unanimity} if for every candidate $x \in C$, there exists a preference profile $\vec{P} = (P_1, \ldots, P_n) \in \P^n$ with $top(P_i) = x$ for every $i \in [n]$, such that $v(x, \vec{P}) \geq 1 - \eps$. 
\end{definition}

$\eps$-super-weak unanimity is a very weak property that all reasonable voting rules should have. We now define what it means for two voting rules to be \emph{$\eps$-close} to each other. 

\begin{definition}
Let $v,v': \P^n \to \Delta(\C)$ be two randomized voting rules. We say that $v$ is \emph{$\eps$-close} to $v'$ if for every preference profile $\vec{P}$ and every candidate $x \in \C$, we have $|v(x,\vec{P}) - v'(x,\vec{P})| \leq \eps$. 
\end{definition}

We now formally state our theorem.

\begin{theorem}
\label{thm:impossibility}
Suppose there are at least three candidates in $\C$, i.e., $|\C| \geq 3$. Let $v: \P^n \to \Delta(\C)$ be any anonymous randomized voting rule that is weakly strategy-proof and satisfies $\eps$-super-weak unanimity. Then, $v$ is $O(\eps)$-close to the random dictatorship voting rule.
\end{theorem}

Theorem \ref{thm:impossibility} is similar to McLennan's impossibility result \cite{McL11} except that we assume $\eps$-super-weak unanimity instead of Pareto efficiency (which is stronger than $\eps$-super-weak unanimity), and our conclusion is that $v$ is $O(\eps)$-close to random dictatorship instead of being equal to random dictatorship. In Theorem \ref{thm:impossibility}, if we consider the special case where $\eps = 0$, we get a theorem that implies McLennan's impossibility result.

Let us briefly mention some aspects of our proof. Our proof uses some tools from McLennan's impossibility result \cite{McL11} and Gibbard's generalization of the Gibbard-Satterthwaite theorem \cite{Gib77}. However, our result does not follow from generalizing McLennan's proof. For example, using McLennan's proof, it is not clear how one can weaken the assumption of Pareto efficiency to $\eps$-Pareto efficiency and still show that the voting rule is close to random dictatorship. By adding ``error terms'' at various places in McLennan's proof, it is not too difficult to show that the voting rule $v$ is $O(\eps n)$-close to random dictatorship. However, if $\eps n$ is large (e.g., $\eps n = \Omega(1)$), then the result is not meaningful. In particular, such a result does not prevent the possibility of constructing a voting rule that is $1/n$-Pareto efficient and weakly strategy-proof. Our impossibility result shows that the voting rule $v$ is $O(\eps)$-close to random dictatorship, as opposed to $O(\eps n)$-close; in order to prove this, we needed to use different analyses and go through substantially more work. In our proof, we do not use an $\eps$-Pareto efficiency assumption because $\eps$-super-weak unanimity is already sufficient; in particular, we show that if the voting rule is weakly strategy-proof, then $\eps$-super-weak unanimity implies $\eps$-strong unanimity, and $\eps$-strong unanimity is what we need in our proof. Thus, we only need to assume that the voting rule satisfies $\eps$-super-weak unanimity instead of the stronger $\eps$-Pareto efficiency or $\eps$-strong unanimity.

We will now prove Theorem \ref{thm:impossibility}. We will prove a sequence of lemmas and claims that describe properties that the voting rule $v$ must have, which will be used to show that $v$ is $O(\eps)$-close to the random dictatorship voting rule.

We first establish some notation and terminology that will be used later; most of the notation and terminology comes from \cite{Gib77}, which is Gibbard's generalization of the Gibbard-Satterthwaite theorem to randomized voting rules. Given a preference ordering $P$ and a pair of candidates $x, y \in \C$, we shall write $x \pref{P!} y$ to mean that $x$ is directly on top of $y$ in the preference ordering $P$, i.e., $x \pref{P} y$ and for every candidate $z \in \C \setminus \{x,y\}$, we have $z \pref{P} x$ if and only if $z \pref{P} y$. Given a preference ordering $P$ and a candidate $y \in \C$, let $P^y$ be the preference ordering $P$ except that $y$ is swapped with the candidate directly above $y$, if such a candidate exists. 

A voting rule $v: \P^n \to \Delta(\C)$ is said to be \emph{pairwise responsive} if for every preference profile $\vec{P} = (P_1, \ldots, P_n) \in \P^n$, every $i \in [n]$, and every pair of candidates $x,y \in \C$ such that $x \pref{P_i}! y$, we have $v(z,(\vec{P}_{-i},{P_i}^y)) = v(z,\vec{P})$ for every candidate $z \in \C \setminus \{x,y\}$. Intuitively, a voting rule is pairwise responsive if for any preference profile, if we take a voter's preference ordering and swap two adjacent candidates, then only the selection probabilities of the swapped candidates can possibly change; the selection probabilities of the other candidates (not involved in the swap) remain the same. This implies that if a voting rule is pairwise responsive and we swap two adjacent candidates in a voter's preference ordering, then the change in the selection probability of one of the swapped candidates is the exact opposite (i.e., the additive inverse) of the change in the selection probability of the other swapped candidate; this is because the selection probabilities of all the candidates must sum up to 1. 

A voting rule $v: \P^n \to \Delta(\C)$ is said to be \emph{pairwise isolated} if for every $x, y \in \C$, every $i \in [n]$, and every $\vec{P} = (P_1, \ldots, P_n), \vec{P'} = (P'_1, \ldots, P'_n) \in \P^n$ such that $x \pref{P_i}! y$, $P'_i = P_i$, and the relative ordering of $x$ and $y$ for $P_j$ is the same as that for $P'_j$ for every $j \in [n]$, we have $v(y,(\vec{P'}_{-i},{P'_i}^y)) - v(y,\vec{P'}) = v(y,(\vec{P}_{-i},{P_i}^y)) - v(y,\vec{P})$. Intuitively, a voting rule is pairwise isolated if for any preference profile, if we take a voter $i$'s preference ordering and swap two adjacent candidates $x$ and $y$ with $x$ being on top of $y$, then the change in the selection probability of $y$ only depends on voter $i$'s preference ordering and the relative ordering of $x$ and $y$ in the other voters' preference orderings.

It is already known that anonymous randomized weakly strategy-proof voting rules are both pairwise responsive and pairwise isolated (see \cite{McL11}). 

\begin{lemma}[\cite{McL11}]
\label{lem:basicProperties}
Let $v: \P^n \to \Delta(\C)$ be any anonymous randomized voting rule that is weakly strategy-proof. Then, $v$ is pairwise responsive and pairwise isolated. 
\end{lemma}

\begin{proof}
This immediately follows from Lemmas 1, 2, and 3 in \cite{McL11}, where Lemmas 2 and 3 in \cite{McL11} follow immediately from Lemmas 1 and 3 in \cite{Gib77}.
\end{proof}

We will use Lemma \ref{lem:basicProperties} at various places below. We now show that if a voting rule is weakly strategy-proof and satisfies $\eps$-super-weak unanimity, then it also satisfies $\eps$-strong unanimity.

\begin{lemma}
\label{lem:superWeakToStrong}
Let $v: \P^n \to \Delta(\C)$ be any anonymous randomized voting rule that is weakly strategy-proof and satisfies $\eps$-super-weak unanimity. Then, $v$ satisfies $\eps$-strong unanimity.
\end{lemma}

The proof of Lemma \ref{lem:superWeakToStrong} uses the pairwise responsive property of $v$ (Lemma \ref{lem:basicProperties}) to obtain $\eps$-strong unanimity from $\eps$-super-weak unanimity. 

\begin{proof}[Proof of Lemma \ref{lem:superWeakToStrong}]
Let $x \in \C$, and let $\vec{P} = (P_1, \ldots, P_n) \in \P^n$ such that $top(P_i) = x$ for every $i \in [n]$. We will show that $v(x,\vec{P}) \geq 1 - \eps$. Since $v$ satisfies $\eps$-super-weak unanimity, there exists a preference profile $\vec{P'} = (P'_1, \ldots, P'_n) \in \P^n$ with $top(P'_i) = x$ for every $i \in [n]$, such that $v(x,\vec{P'}) \geq 1-\eps$. Now, we observe that since candidate $x$ is at the top for every voter in both $\vec{P}$ and $\vec{P'}$, we can obtain $\vec{P'}$ from $\vec{P}$ by performing a sequence of swaps of adjacent candidates (in the voters' preference orderings) such that none of the swaps involve candidate $x$. Since $v$ is pairwise responsive (Lemma \ref{lem:basicProperties}), the selection probability of $x$ is not changed by any of the swaps. Thus, we have $v(x,\vec{P}) = v(x,\vec{P'}) \geq 1-\eps$, as required.
\end{proof}

Given a preference profile $\vec{P} = (P_1, \ldots, P_n) \in \P^n$, let $top(\vec{P}) = (top(P_1), \ldots, top(P_n))$. We now show that if a voting rule is weakly strategy-proof and satisfies $\eps$-super-weak unanimity, then the voting rule is ``close to'' being ``tops-only'', i.e., depending only on the vector $top(\vec{P})$ of top choices of the preference orderings in the input $\vec{P}$; the ordering of the candidates below the top choices only affect the probabilities of the voting rule by a small amount $O(\eps)$. 

\begin{lemma}
\label{lem:closeToTopsOnly}
Let $v: \P^n \to \Delta(\C)$ be any anonymous randomized voting rule that is weakly strategy-proof and satisfies $\eps$-super-weak unanimity. Then, for every pair of preference profiles $\vec{P}, \vec{P'} \in \P^n$ such that $top(\vec{P}) = top(\vec{P'})$, we have $|v(x,\vec{P}) - v(x,\vec{P'})| \leq m\eps$ for every $x \in \C$.
\end{lemma}

The proof of Lemma \ref{lem:closeToTopsOnly} uses the pairwise isolation property (Lemma \ref{lem:basicProperties}) and the $\eps$-strong unanimity property (Lemma \ref{lem:superWeakToStrong}); the greatest difficulty in the proof lies in ensuring that the error (the $m\eps$ in Lemma \ref{lem:closeToTopsOnly}) is $O(\eps)$ instead of $O(\eps n)$, which is much too large. 

\begin{proof}[Proof of Lemma \ref{lem:closeToTopsOnly}]
Let $\C = \{a_1, \ldots, a_m\}$, let $\vec{P} = (P_1, \ldots, P_n), \vec{P'} = (P'_1, \ldots, P'_n) \in \P^n$ such that $top(\vec{P}) = top(\vec{P'})$, and let $x \in \C$. For $j = 0, \ldots, m$, let $\vec{P}^{(j)} = ({P^{(j)}}_1, \ldots, {P^{(j)}}_n)$ be defined as follows: for each $i \in [n]$, let ${P^{(j)}}_i = P'_i$ if $top(P_i) \in \{a_1, \ldots, a_j\}$, and let ${P^{(j)}}_i = P_i$ otherwise. We note that $\vec{P}^{(0)} = \vec{P}$ and $\vec{P}^{(m)} = \vec{P'}$. We will show that for every $0 \leq j \leq m-1$, we have $|v(x,\vec{P}^{(j)}) - v(x,\vec{P}^{(j+1)})| \leq \eps$. The lemma then immediately follows by repeated application of the triangle inequality.

Let $0 \leq j \leq m-1$. From the definition of $\vec{P}^{(j)}$ and $\vec{P}^{(j+1)}$, we see that $\vec{P}^{(j)}$ and $\vec{P}^{(j+1)}$ only differ in the components $i \in [n]$ for which $top(P_i) = a_{j+1}$; let $I$ be the set of all such $i \in [n]$. For each $i \in I$, we have ${\vec{P}^{(j)}}_i = P_i$ and ${\vec{P}^{(j+1)}}_i = P'_i$. However, since $top(\vec{P}) = top(\vec{P'})$, $top(P_i) = a_{j+1}$ is at the top of both ${\vec{P}^{(j)}}_i$ and ${\vec{P}^{(j+1)}}_i$ for every $i \in I$, so $\vec{P}^{(j+1)}$ can be obtained from $\vec{P}^{(j)}$ via a sequence of swaps of adjacent candidates (in the preference orderings of voters in $I$) that do not involve candidate $a_{j+1}$. Let $\vec{Q}^{(0)}, \ldots, \vec{Q}^{(r)}$ be any sequence of preference profiles generated by such a sequence of swaps, where $\vec{Q}^{(0)} = \vec{P}^{(j)}$ and $\vec{Q}^{(r)} = \vec{P}^{(j+1)}$. Now, we observe that
\begin{align*}
v(x,\vec{P}^{(j+1)}) - v(x,\vec{P}^{(j)}) = \sum_{k=0}^{r-1} \left( v(x,\vec{Q}^{(k+1)}) - v(x,\vec{Q}^{(k)}) \right). \tag{1}
\end{align*}
For each $0 \leq k \leq r$, let $\vec{Q'}^{(k)}$ be the same as $\vec{Q}^{(k)}$ except that for every $i \in [n] \setminus I$, the candidate $a_{j+1}$ in the preference ordering ${\vec{Q}^{(k)}}_i$ is moved to the top of the preference ordering; we note that $a_{j+1}$ is at the top of all the preference orderings in $\vec{Q'}^{(k)}$. We will now show that for each $0 \leq k \leq r-1$, we have
\begin{align*}
v(x,\vec{Q}^{(k+1)}) - v(x,\vec{Q}^{(k)}) = v(x,\vec{Q'}^{(k+1)}) - v(x,\vec{Q'}^{(k)}). \tag{2}
\end{align*}
Let $0 \leq k \leq r-1$. Recall that $\vec{Q}^{(k+1)}$ is obtained from $\vec{Q}^{(k)}$ by performing a swap of adjacent candidates in the preference ${\vec{Q}^{(k)}}_i$ for some $i \in I$, and the swap does not involve candidate $a_{j+1}$. If the swap does not involve the candidate $x$, then since $v$ is pairwise responsive (Lemma \ref{lem:basicProperties}), both sides of (2) are 0, which proves (2). Thus, we now assume that the swap involves the candidate $x$. If the swap moves $x$ upwards, then since $v$ is pairwise isolated (Lemma \ref{lem:basicProperties}), (2) holds. If the swap moves $x$ downwards, then the swap moves some other candidate upwards, say, candidate $y$. Then, since $v$ is pairwise responsive and the selection probabilities of the candidates must add up to 1, we have $v(x,\vec{Q}^{(k+1)}) - v(x,\vec{Q}^{(k)}) = -( v(y,\vec{Q}^{(k+1)}) - v(y,\vec{Q}^{(k)}) )$. Since $v$ is pairwise isolated, we have $v(y,\vec{Q}^{(k+1)}) - v(y,\vec{Q}^{(k)}) = v(y,\vec{Q'}^{(k+1)}) - v(y,\vec{Q'}^{(k)}) = - (v(x,\vec{Q'}^{(k+1)}) - v(x,\vec{Q'}^{(k)}))$, where the last equality again uses the fact that $v$ is pairwise responsive. Combining the above equations yields (2), as required. Thus, we have shown (2).

Now, combining (1) and (2), we get
\begin{align*}
v(x,\vec{P}^{(j+1)}) - v(x,\vec{P}^{(j)}) = \sum_{k=0}^{r-1} \left( v(x,\vec{Q'}^{(k+1)}) - v(x,\vec{Q'}^{(k)}) \right) = v(x,\vec{Q'}^{(r)}) - v(x,\vec{Q'}^{(0)}).
\end{align*}
To complete the proof of the lemma, it suffices to show that $|v(x,\vec{Q'}^{(r)}) - v(x,\vec{Q'}^{(0)})| \leq \eps$. 

Since $v$ satisfies $\eps$-super-weak unanimity, we have that by Lemma \ref{lem:superWeakToStrong}, $v$ also satisfies $\eps$-strong unanimity. Thus, we have $v(a_{j+1},\vec{Q'}^{(r)}) \geq 1-\eps$ and $v(a_{j+1},\vec{Q'}^{(0)}) \geq 1-\eps$, since candidate $a_{j+1}$ is at the top of all the preference orderings in $\vec{Q'}^{(r)}$ and $\vec{Q'}^{(0)}$.

If $x = a_{j+1}$, then $|v(x,\vec{Q'}^{(r)}) - v(x,\vec{Q'}^{(0)})| = |v(a_{j+1},\vec{Q'}^{(r)}) - v(a_{j+1},\vec{Q'}^{(0)})| \leq \eps$, since $v(a_{j+1},\vec{Q'}^{(r)}) \in [1-\eps,1]$ and $v(a_{j+1},\vec{Q'}^{(0)}) \in [1-\eps,1]$. 

If $x \neq a_{j+1}$, then $v(x,\vec{Q'}^{(r)}) \in [0,\eps]$ and $v(x,\vec{Q'}^{(0)}) \in [0,\eps]$ (since $v(a_{j+1},\vec{Q'}^{(r)}) \geq 1-\eps$ and $v(a_{j+1},\vec{Q'}^{(0)}) \geq 1-\eps$), so $|v(x,\vec{Q'}^{(r)}) - v(x,\vec{Q'}^{(0)})| \leq \eps$, as required. 

This completes the proof of the lemma.
\end{proof}

We now show that if a voting rule is weakly strategy-proof and satisfies $\eps$-super-weak unanimity, then the selection probability of any candidate $x$ is ``close to'' depending only on the number of voters with $x$ as the top choice; the top choices of the other voters only affect the selection probability of $x$ by a small amount $O(\eps)$. 

\begin{lemma}
\label{lem:timesAtTop}
Let $v: \P^n \to \Delta(\C)$ be any anonymous randomized voting rule that is weakly strategy-proof and satisfies $\eps$-super-weak unanimity. Then, for every candidate $x \in \C$, and every pair of preference profiles $\vec{P} = (P_1, \ldots, P_n), \vec{P'} = (P'_1, \ldots, P'_n) \in \P^n$ such that $|\{i \in [n] : top(P_i) = x\}| = |\{i \in [n] : top(P'_i) = x\}|$, we have $|v(x,\vec{P}) - v(x,\vec{P'})| \leq 2m\eps$.
\end{lemma}

\begin{proof}
Let $x \in \C$, and let $\vec{P} = (P_1, \ldots, P_n), \vec{P'} = (P'_1, \ldots, P'_n) \in \P^n$ such that $|\{i \in [n] : top(P_i) = x\}| = |\{i \in [n] : top(P'_i) = x\}|$. Let $k = |\{i \in [n] : top(P_i) = x\}|$. Since $v$ is anonymous, we can assume without loss of generality that for $i = 1, \ldots, k$, we have $top(P_i) = x$ and $top(P'_i) = x$. Then, for $i = k+1, \ldots, n$, we have $top(P_i) \neq x$ and $top(P'_i) \neq x$. Let $a_1, \ldots, a_m$ be any ordering of the set of candidates $\C$ such that $a_m = x$, where $a_1$ is the highest-ranked candidate and $a_m$ is the lowest-ranked candidate. Let $\vec{P^*} = (P^*_1, \ldots, P^*_n)$ be the preference profile defined as follows: for $i = 1, \ldots, k$, let $top(P^*_i) = x$ and the rest of $P^*_i$ is ordered according to the ordering $a_1, \ldots, a_m$; for $i = k+1, \ldots, n$, let $P^*_i$ be the ordering $a_1, \ldots, a_m$. 

We first show that $|v(x,\vec{P}) - v(x,\vec{P^*})| \leq m\eps$. We will perform a sequence of operations on $\vec{P}$ to obtain $\vec{P^*}$, and we will analyze how these operations affect the selection probability $v(x,\cdot)$ of $x$. We start with $\vec{P}$, and for every $i \in \{k+1, \ldots, n\}$, we simultaneously move the candidate $x$ in $P_i$ to the bottom. Since $top(P_i) \neq x$ for $i = k+1, \ldots, n$, by Lemma \ref{lem:closeToTopsOnly}, this operation changes the selection probability of $x$ by at most $m\eps$. Now, we observe that from this new preference profile, we can obtain $\vec{P^*}$ by performing a sequence of swaps of adjacent candidates in the preference orderings such that none of the swaps involve candidate $x$. Since $v$ is pairwise responsive (Lemma \ref{lem:basicProperties}), none of these swaps change the selection probability of $x$. Thus, the overall change in the selection probability of $x$ is at most $m\eps$, so $|v(x,\vec{P}) - v(x,\vec{P^*})| \leq m\eps$, as required.

By a similar argument, we also have $|v(x,\vec{P'}) - v(x,\vec{P^*})| \leq m\eps$. Thus, by the triangle inequality, we have $|v(x,\vec{P}) - v(x,\vec{P'})| \leq 2m\eps$. This completes the proof of the lemma.
\end{proof} 

We now use the above lemmas to prove Theorem \ref{thm:impossibility}. Suppose $|\C| \geq 3$. Fix $v: \P^n \to \Delta(\C)$ to be any anonymous randomized voting rule that is weakly strategy-proof and satisfies $\eps$-super-weak unanimity. By Lemma \ref{lem:basicProperties}, $v$ is pairwise responsive and pairwise isolated, and by Lemma \ref{lem:superWeakToStrong}, $v$ also satisfies $\eps$-strong unanimity. We know from Lemma \ref{lem:timesAtTop} that the selection probability of any candidate $x$ is close to depending only on the number of voters with $x$ as the top choice. Thus, we will define a function $v': \C \times \{0, \ldots, n\}$ that, on input a candidate $x \in \C$ and a number $j \in \{0, \ldots, n\}$, specifies the approximate selection probability of $x$ when exactly $j$ voters have $x$ as their top choice. 

Let $\C = \{a_1, \ldots, a_m\}$, where $a_1, \ldots, a_m$ is any fixed ordering of the candidates in $\C$. Let $v': \C \times \{0, \ldots, n\} \to [0,1]$ be defined by $v'(x,j) = v(x,\vec{P}^{x,j})$, where $\vec{P}^{x,j} = (P_1, \ldots, P_n)$ is the preference profile defined as follows: for $i = 1, \ldots, j$, let the top choice of $P_i$ be candidate $x$ and let the other candidates be ordered according to the ordering $a_1, \ldots, a_m$; for $i = j+1, \ldots n$, let $P_i$ be the ordering $a_1, \ldots, a_m$ with candidate $x$ moved to the bottom. By Lemma \ref{lem:timesAtTop}, the following claim follows immediately.

\begin{claim}
\label{claim:vCloseTov'}
For every preference profile $\vec{P} = (P_1, \ldots, P_n)$ with $j := |\{i \in [n] : top(P_i) = x\}|$, and every candidate $x \in \C$, we have $|v(x,\vec{P}) - v'(x,j)| \leq 2m\eps$.
\end{claim}

We now show that the candidates are ``close to being anonymous'' in the sense that changing the candidate in the input for $v'$ only changes the value of $v'$ by a small amount $O(\eps)$.

\begin{claim}
\label{claim:candidatesAreAnonymous}
For every pair of candidates $x,y \in \C$ and every $j \in \{0, \ldots, n\}$, we have $|v'(x,j) - v'(y,j)| \leq 14m\eps$.
\end{claim}

\begin{proof}
Let $x,y \in \C$ and $j \in [n]$. Let $z$ be any candidate in $\C \setminus \{x,y\}$. Let $\vec{P} = (P_1, \ldots, P_n)$ be any preference profile such that for every $i \in [n]$, we have $top(P_i) = z$ and candidate $x$ is directly below $z$. Now, for $i = 1, \ldots, j$, we swap the candidates $z$ and $x$ in $P_i$ so that $x$ is now at the top; let's call the resulting preference profile $\vec{P'}$. Since the swaps we performed only involved the candidates $z$ and $x$, and since $v$ is pairwise responsive, we have $v(x,\vec{P'}) - v(x,\vec{P}) = - (v(z,\vec{P'}) - v(z,\vec{P})) = v'(z,n) - v'(z,n-j) + \delta$, where $|\delta| \leq 4m\eps$ (by Claim \ref{claim:vCloseTov'}). Now, we note that $|v(x,\vec{P})| \leq \eps$ (by strong unanimity of $v$), so $v(x,\vec{P'}) = v'(z,n) - v'(z,n-j) + \delta'$, where $|\delta'| \leq 5m\eps$. Thus, by Claim \ref{claim:vCloseTov'}, we have $v'(x,j) = v'(z,n) - v'(z,n-j) + \gamma$, where $|\gamma| \leq 7m\eps$.

By a similar argument but where we use $y$ instead of $x$, we also have $v'(y,j) = v'(z,n) - v'(z,n-j) + \gamma'$, where $|\gamma'| \leq 7m\eps$. Thus, we have
\begin{align*}
|v'(x,j) - v'(y,j)| = |\gamma - \gamma'| \leq 14m\eps.
\end{align*}
This completes the proof of the claim.
\end{proof} 

We now show that for any candidate $x \in \C$, the function $v'(x,\cdot)$ is ``close to being linear'' in the following sense: When the number of top choice votes for candidate $x$ (i.e., the $j$ in $v'(x,j)$) is increased by $\ell$, the change in $v'(x,\cdot)$ depends very little on how many top choice votes candidate $x$ had initially, which can only affect the change in $v'(x,\cdot)$ by a small amount $O(\eps)$. This roughly means that as the number of top choice votes for candidate $x$ increases, the selection probability of $x$ increases roughly linearly. 

\begin{claim}
\label{claim:slidingWindow}
Let $x \in \C$, let $j,j' \in \{0, \ldots, n-1\}$, and let $\ell \in [n]$ such that $j + \ell \leq n$ and $j' + \ell \leq n$. Then,
\begin{align*}
v'(x,j+\ell) - v'(x,j) = v'(x,j'+\ell) - v'(x,j') + \delta
\end{align*}
for some $\delta \in \R$ such that $|\delta| \leq 64m\eps$.
\end{claim}

\begin{proof}
We first show that $v'(x,j+\ell) - v'(x,j) = v'(x,\ell) - v'(x,0) + \gamma$, where $|\gamma| \leq 32m\eps$. Let $y,z \in \C$ such that $x,y,z$ are all distinct. Let $\vec{P} = (P_1, \ldots, P_n)$ be any preference profile such that for $i = 1, \ldots, j$, we have $top(P_i) = x$, and for $i = j+1, \ldots, j+\ell$, we have $top(P_i) = y$ and candidate $x$ is directly below $y$, and for $i = j+\ell+1, \ldots, n$, we have $top(P_i) = z$. Now, for $i = j+1, \ldots, j+\ell$, we swap the candidates $y$ and $x$ in $P_i$ so that $x$ is now at the top; let's call the resulting preference profile $\vec{P'}$. Since the swaps we performed only involved the candidates $y$ and $x$, and since $v$ is pairwise responsive, we have $v(x,\vec{P'}) - v(x,\vec{P}) = - (v(y,\vec{P'}) - v(y,\vec{P})) = v'(y,\ell) - v'(y,0) + \delta'$, where $|\delta'| \leq 4m\eps$. Now, by Claim \ref{claim:candidatesAreAnonymous}, we have $|v'(y,\ell) - v'(x,\ell)| \leq 14m\eps$ and $|v'(y,0) - v'(x,0)| \leq 14m\eps$, so $v(x,\vec{P'}) - v(x,\vec{P}) = v'(x,\ell) - v'(x,0) + \gamma$, where $|\gamma| \leq 32m\eps$. 

By a similar argument, we also have $v'(x,j'+\ell) - v'(x,j') = v'(x,\ell) - v'(x,0) + \gamma'$, where $|\gamma'| \leq 32m\eps$. The claim then follows by the triangle inequality.
\end{proof}

Using Claim \ref{claim:slidingWindow}, we now show that $v'(x,j)$ (which approximates the selection probability of candidate $x$ with $j$ top choice votes) is close to $\frac{j}{n}$, which is the selection probability of $x$ for the random dictatorship voting rule. Even though Claim \ref{claim:slidingWindow} says that $v'(x,j)$ is close to being linear, naive usage of Claim \ref{claim:slidingWindow} would result in $O(\eps n)$ error, which is too high; however, by using a better approach, we can make the error only $O(\eps)$.

\begin{claim}
\label{claim:closeToRandomDictatorship}
Let $x \in \C$. For every $j \in \{0, \ldots, n\}$, we have $|v'(x,j) - \frac{j}{n}| \leq O(\eps)$. 
\end{claim}

\begin{proof}
Since $v$ satisfies $\eps$-strong unanimity, we have $v'(x,n) \geq 1 - \eps$ and $v'(x,0) \leq \eps$. Now, for every $q \leq n/2$, we have $v'(x,2q) - v'(x,q) = v'(x,q) - v'(x,0) + O(\eps)$ by Claim \ref{claim:slidingWindow}, so 
\begin{align*}
v'(x,2q) = 2v'(x,q) + O(\eps). \tag{1}
\end{align*}
Let $q \in \{0, \ldots, n\}$ such that $r_q := |v'(x,q)-q/n|$ is maximal. We will show that $r_q \leq O(\eps)$.

Case 1: $q \leq n/2$. By (1), we have
\begin{align*}
r_{2q} = |v'(x,2q) - 2q/n| = |2v'(x,q)+O(\eps) - 2q/n| = 2r_q - O(\eps).
\end{align*}
By the maximality of $r_q$, we have $r_{2q} \leq r_q$, so $2r_q - O(\eps) \leq r_q$, so $r_q \leq O(\eps)$, as required.

Case 2: $q > n/2$. Let $q' = n - q$. Then, by Claim \ref{claim:slidingWindow}, we have $v'(x,q')-v'(x,0) = v'(x,n) - v'(x,q) + O(\eps)$, so $v'(x,q') = 1 - v'(x,q) + O(\eps)$. Then, we have 
\begin{align*}
r_{q'} = |v'(x,q')-q'/n| = |1 - v'(x,q) + O(\eps) - (n-q)/n| = r_{q} - O(\eps). 
\end{align*}
Now, by (1), we have
\begin{align*}
r_{2q'} = |v'(x,2q') - 2q'/n| = |2v'(x,q') + O(\eps) - 2q'/n| = 2r_{q'} - O(\eps) = 2r_q - O(\eps).
\end{align*}
By the maximality of $r_q$, we have $r_{2q'} \leq r_q$, so $2r_q - O(\eps) \leq r_q$, so $r_q \leq O(\eps)$, as required.
\end{proof}

We are now ready to complete the proof of Theorem \ref{thm:impossibility}, i.e., we will show that $v$ is $O(\eps)$-close to the random dictatorship voting rule. Let $\vec{P} = (P_1, \ldots, P_n) \in \P^n$, let $x \in \C$, and let $j = |\{i \in [n] : top(P_i) = x\}|$. We will show that $|v(x,\vec{P}) - \frac{j}{n}| \leq O(\eps)$. By Claim \ref{claim:vCloseTov'}, we have $|v(x,\vec{P}) - v'(x,j)| \leq O(\eps)$, and by Claim \ref{claim:closeToRandomDictatorship}, we also have $|v'(x,j) - \frac{j}{n}| \leq O(\eps)$. Thus, by the triangle inequality, we have $|v(x,\vec{P}) - \frac{j}{n}| \leq O(\eps)$, as required. This completes the proof of Theorem \ref{thm:impossibility}.

\section{Acknowledgments}

We thank Ron Rivest for helpful discussions. We also thank anonymous reviewers for helpful suggestions and comments.

\bibliographystyle{amsalpha}
\bibliography{ImpossibilityOfVoting}

\appendix
\renewcommand\thesection{Appendix \Alph{section}}

\section{Background Information on the Gibbard-Satterthwaite Theorem}
\label{app:Gibbard-Satterthwaite}

Roughly speaking, a voting rule is said to be \emph{strategy-proof} if no voter can gain utility in expectation by lying about her true preferences. We now give the formal definition of strategy-proof.

\begin{definition}[Strategy-proof]
A voting rule $v: \P^n \to \Delta(\C)$ is \emph{strategy-proof} if for every $i \in [n]$, every preference profile $\vec{P}_{-i} \in \P^{n-1}$, every pair of preference orderings $P_i, P_i' \in \P$, and every utility function $u_i$ that is consistent with $P_i$, we have
\begin{align*}
  \E[u_i(v(\vec{P}_{-i},P_i))] \geq \E[u_i(v(\vec{P}_{-i},P_i'))].
\end{align*}
\end{definition} 

It is desirable for a voting rule to be strategy-proof, since we can then expect voters to honestly submit their true preferences, and thus the candidate chosen by the voting rule will better reflect the voters' true preferences. Unfortunately, if there are at least three candidates, then it is not possible for a deterministic and onto voting rule to be strategy-proof unless it is \emph{dictatorial}, i.e., there exists some voter $i$ such that the voting rule simply always chooses voter $i$'s top choice. This was shown independently by Gibbard \cite{Gib73} and Satterthwaite \cite{Sat75}, and is known as the Gibbard-Satterthwaite theorem.

\begin{theorem}[Gibbard-Satterthwaite \cite{Gib73,Sat75}]
Suppose there are at least three candidates, i.e., $|\C| \geq 3$. Let $v: \P^n \to \C$ be any deterministic voting rule that is onto and strategy-proof. Then, $v$ is dictatorial, i.e., there exists an $i \in [n]$ such that $v(P_1, \ldots, P_n) = top(P_i)$ for every preference profile $(P_1, \ldots, P_n) \in \P^n$. 
\end{theorem}

The Gibbard-Satterthwaite theorem considers voting rules that are \emph{deterministic}. However, several years later, Gibbard \cite{Gib77} generalized the Gibbard-Satterthwaite theorem to \emph{randomized} voting rules. Before we state Gibbard's generalized impossibility result, let us state some required definitions. A (randomized) voting rule $v: \P^n \to \Delta(\C)$ is said to be \emph{unilateral} if it only depends on the preference of a single voter, i.e., there exists an $i \in [n]$ such that $v(\vec{P}) = v(\vec{P'})$ for every $\vec{P} = (P_1, \ldots, P_n), \vec{P'} = (P_1', \ldots, P_n') \in \P^n$ such that $P_i = P_i'$. A (randomized) voting rule $v: \P^n \to \Delta(\C)$ is said to be \emph{duple} if $v$ always chooses some candidate from a fixed set of two candidates, i.e., there exist candidates $x,y \in \C$ such that $v(z,\vec{P}) = 0$ for every $z \in \C \setminus \{x,y\}$ and $\vec{P} \in \P^n$. 

Intuitively, when there are at least three candidates, both unilateral rules and duple rules are undesirable, since the former only consider a single voter's preference, and the latter essentially ignore all but two candidates. Gibbard's generalized impossibility result \cite{Gib77} states that any randomized strategy-proof voting rule is a probability distribution over unilateral rules and duple rules.

\begin{theorem}[Gibbard \cite{Gib77}]
Let $v: \P^n \to \Delta(\C)$ be any randomized voting rule that is strategy-proof. Then, $v$ is a distribution over unilateral rules and duple rules, i.e., there exist randomized voting rules $v_1, \ldots, v_t$ and weights $\alpha_1, \ldots, \alpha_t \in (0,1]$ with $\sum_{i=1}^t \alpha_i = 1$, such that each $v_i$ is unilateral or duple, and $v(x,\vec{P}) = \alpha_1 v_1(x,\vec{P}) + \cdots + \alpha_t v_t(x,\vec{P})$ for every $\vec{P} \in \P^n$ and $x \in \C$.
\end{theorem}

A corollary of Gibbard's impossibility result is that if a randomized voting rule is strategy-proof and Pareto efficient, then it is a probability distribution over dictatorial voting rules.

\begin{corollary}[Gibbard \cite{Gib77}]
Let $v: \P^n \to \Delta(\C)$ be any randomized voting rule that is strategy-proof and Pareto efficient. Then, $v$ is a distribution over dictatorial voting rules, i.e., there exist dictatorial voting rules $v_1, \ldots, v_t$ and weights $\alpha_1, \ldots, \alpha_t \in (0,1]$ with $\sum_{i=1}^t \alpha_i = 1$, such that $v(x,\vec{P}) = \alpha_1 v_1(x,\vec{P}) + \cdots + \alpha_t v_t(x,\vec{P})$ for every $\vec{P} \in \P^n$ and $x \in \C$.
\end{corollary}

\end{document}